\documentclass[letterpaper, 10 pt, conference]{ieeeconf}
\IEEEoverridecommandlockouts
\usepackage[colorlinks=False,allcolors=blue]{hyperref}
\usepackage[margin=1in]{geometry}
\usepackage{subcaption}
\usepackage{booktabs}
\usepackage{zref-user}
\usepackage{footmisc}

\usepackage{wrapfig}

\usepackage{graphicx}
\usepackage{amsmath}
\usepackage{amssymb}
\usepackage{bm}
\usepackage{mathtools}

\def\namedlabel#1#2{\begingroup
    #2%
    \def\@currentlabel{#2}%
    \phantomsection\label{#1}\endgroup
}
\makeatletter
\newcommand{\leqnomode}{\tagsleft@true\let\veqno\@@leqno}
\makeatother
\newcommand\PMPODE[1]{\hyperref[PMPODE]{$\textbf{ODE}_{#1}$}\xspace}

\newcommand\PMPODEg[1]{\hyperref[PMPODEg]{$\textbf{ODE}^g_{#1}$}\xspace}
\newcommand\OCPg[1]{\hyperref[OCPg]{$\textbf{OCP}^g_{#1}$}\xspace}

\newcommand\OCP[1]{\hyperref[OCP]{$\textbf{OCP}_{#1}$}\xspace}
\newcommand\BVP[1]{\hyperref[BVP]{$\textbf{BVP}_{#1}$}\xspace}

\usepackage{tabularx}
\newcolumntype{C}{>{\centering\arraybackslash}p{19mm}}
\newcolumntype{G}{>{\centering\arraybackslash}p{4mm}}
\newcolumntype{S}{>{\centering\arraybackslash\scriptsize}p{4mm}}

\usepackage[font=small,labelfont=bf]{caption}
\usepackage{color}
\usepackage[usenames,dvipsnames]{xcolor}

\usepackage{cleveref} %

\usepackage{xspace}

\let\labelindent\relax
\usepackage{enumitem} %

\usepackage{amssymb}%
\usepackage{algorithm}
\makeatletter
\renewcommand*{\ALG@name}{Alg.}
\makeatother
\usepackage[noend]{algpseudocode}

\usepackage{xpatch} %

\newtheorem{definition}{Definition}

\newtheorem{assumption}{Assumption}
\newtheorem{lemma}{Lemma}
\newtheorem{theorem}{Theorem}
\newtheorem{corollary}{Corollary}
\newtheorem{proposition}{Proposition}

\newcommand\mydots{\hbox to 1em{.\hss.\hss.}}

\providecommand{\R}{\ensuremath \mathbb{R}}

\newcommand{\ie}{\textit{i.e., }}
\newcommand{\eg}{\textit{e.g., }}

\newcommand{\norm}[1]{\left\Vert#1\right\Vert}

\newcommand{\bigO}{\mathcal{O}}

\newcommand{\defeq}{\vcentcolon=}

\newcommand{\E}{\mathbf{E}} 

\DeclareMathOperator*{\argmin}{arg\,min}

\newcommand{\eye}{{\mathbf{I}}}
\newcommand{\Vpca}{\mathbf{V}_{\text{E}}}
\newcommand{\Vopt}{\mathbf{V}_{\text{opt}}}

\newcommand{\fstate}{\mathbf{x}}

\newcommand{\rstate}{\mathbf{x}_{r}}
\newcommand{\rstatebar}{\bar{\mathbf{x}}_{r}}
\newcommand{\nstate}{\mathbf{x}_{n}}
\newcommand{\nstatebar}{\bar{\mathbf{x}}_{n}}

\newcommand{\obs}{\mathbf{y}}

\newcommand{\manifold}{\mathcal{M}}

\newcommand{\fode}{\mathbf{f}}

\newcommand{\fiber}{f^s}
\newcommand{\foliation}[1]{\mathcal{W}_{\mathrm{loc}}^s(#1)}
\newcommand{\localfoliation}[1]{\mathcal{W}_{\mathrm{loc}}^s(#1)}

\newcommand{\wmap}{\mathbf{w}} 
\newcommand{\parammap}{\bm{\omega}} 
\newcommand{\vmap}{\bm{\nu}} 

\newcommand{\chart}{\bm{\sigma}}
\newcommand{\normalchart}{\bm{\tau}}

\newcommand{\normalbundleeps}{\mathbf{N}^{s, \varepsilon}}
\newcommand{\Vconst}{\mathbf{V}_{0}}
\newcommand{\Vlin}{\mathbf{V}_{1}}
\newcommand{\graphfunc}{\mathbf{g}}

\newcommand{\p}{\mathbf{p}}

\newcommand{\ctrl}{\mathbf{u}}

\newcommand{\spectralsubspace}{E}

\newcommand{\zero}{\mathbf{0}}

\newcommand{\flow}{\mathbf{F}^t}

\newcommand{\tubular}{\mathcal{T}^\varepsilon}

\newcommand{\nfstate}{{n_{f}}}
\newcommand{\nrstate}{n}
\newcommand{\nctrl}{m}

\newcommand{\nsmooth}{r}

\newcommand{\jac}{\mathbf{D}}

\newcommand{\pp}{\mathbf{p}}
\newcommand{\qq}{\mathbf{x}}
\newcommand{\vv}{\mathbf{v}}

\newtheorem{preremark3}{Theorem}[section]

\usepackage{mdframed}
\usepackage{lipsum}
\newmdtheoremenv{theo}{Theorem}

\newcommand{\fplus}[1][black]{%
  \begingroup\leavevmode\color{#1}%
  \setlength{\unitlength}{0.6em}%
  \linethickness{.15em}%
  \begin{picture}(1,1)
  \put(0,0.5){\line(1,0){1}}
  \put(0.5,0){\line(0,1){1}}
  \end{picture}%
  \endgroup
}

\usepackage{multirow}

\usepackage[backend=biber,style=ieee,url=false,doi=false,eprint=false, date=year]{biblatex}
\AtEveryBibitem{\clearfield{pages}}
\AtEveryBibitem{\clearfield{volume}}
\AtEveryBibitem{\clearfield{number}}

\usepackage{etoolbox}
\newtoggle{ext}
\usepackage{hyperref}
\toggletrue{ext}  

\title{\LARGE \bf
Taming High-Dimensional Dynamics: \\ Learning Optimal Projections onto Spectral Submanifolds 
\vspace{-4mm}
}

\author{Hugo Buurmeijer$^{1}$, Luis Pabon$^{1}$, John Irvin Alora$^{1}$, Roshan S. Kaundinya$^{2}$, \\ George Haller$^{2}$, Marco Pavone$^{1,3}$%
\thanks{{\footnotesize$^{1}$Department of Aeronautics \& Astronautics, Stanford University.
        {\tt\scriptsize \{hbuurmei,lpabon,jjalora,pavone\}@stanford.edu}}}%
\thanks{{\footnotesize$^{2}$Institute of Mechanical Systems, ETH Z\"urich. \phantom{abcdefghijklmn}  \hbox{{\tt\scriptsize \{sroshan,georgehaller\}@ethz.ch}}}}%
\thanks{{\footnotesize$^{3}$NVIDIA Research. \phantom{abcdefghijklmn}}}%
\vspace{-4mm}
}

\begin{document}

\maketitle
\thispagestyle{empty}
\pagestyle{empty}

\begin{abstract}
High-dimensional nonlinear systems pose considerable challenges for modeling and control across many domains, from fluid mechanics to advanced robotics.
Such systems are typically approximated with reduced-order models, which often rely on orthogonal projections, a simplification that may lead to large prediction errors.
In this work, we derive optimality of fiber-aligned projections onto spectral submanifolds, preserving the nonlinear geometric structure and minimizing long-term prediction error.
We propose a data-driven procedure to learn these projections from trajectories and demonstrate its effectiveness through a 180-dimensional robotic system. Our reduced-order models achieve up to fivefold improvement in trajectory tracking accuracy under model predictive control compared to the state of the art.
\end{abstract}

\section{Introduction}\label{sec:introduction}
We address the challenge of modeling and controlling high-dimensional nonlinear systems, prevalent in fields like fluid mechanics and robotics, where computational complexity necessitates reduced-order models (ROMs). Traditional ROMs, typically based on orthogonal projections like Proper Orthogonal Decomposition (POD) \cite{berkooz2003proper}, often fail to accurately capture nonlinear dynamics critical for reliable long-term prediction and control \cite{cenedese2022data-driven}.

Reducing the dynamics to Spectral Submanifolds (SSMs) offers a rigorous alternative for capturing dominant slow dynamics in nonlinear systems~\cite{haller2016nonlinear}. In general, the system's trajectory will not lie exactly on an SSM, requiring a mapping from off-manifold states onto on-manifold states.
Prior data-driven SSM methods \cite{alora2025discovering, AloraCenedeseEtAl2023, kaundinya2025data-driven, alora2023robust} carry out this mapping via orthogonal projections onto the tangent space of the SSM.
While often sufficient, such orthogonal projections neglect the system’s stable fiber geometry, leading to inaccuracies in capturing transient behaviors and thus limiting ROM performance in predictive control applications.
\begin{figure}[t]
\centering
    \vspace{0.2cm}
    {\includegraphics[width=0.85\columnwidth]{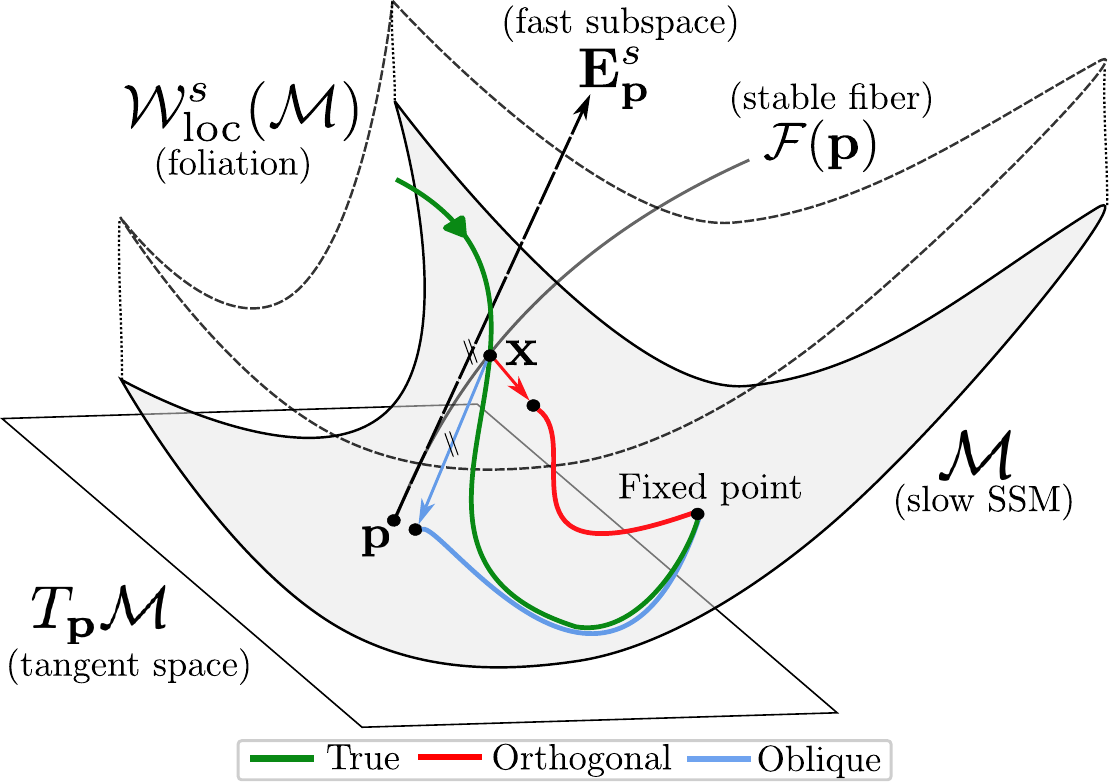}}
    \caption{\textbf{Invariant Foliation of a slow SSM}.
    A slow SSM, $\manifold$, is foliated by stable fibers forming a local stable foliation, $\localfoliation{\manifold}$. 
    The green trajectory represents the evolution of the system from an initial condition off of the SSM. We show how oblique projections (blue) improve prediction accuracy over projecting orthogonally (red) onto the SSM.
    } \label{fig:foliation}
    \vspace{-5mm}
\end{figure}

Motivated by this limitation, we specifically tackle reduced-order modeling for high-dimensional nonlinear control systems of the form
\begin{equation} \label{eq:ODE_intro} \dot{\fstate}(t)=\fode(\fstate(t))+\mathbf{B}(\fstate)\ctrl(t), 
\end{equation}
where $\fstate\in\R^{\nfstate}$ is the state, $\ctrl\in\R^{\nctrl}$ is the control input, and $\fode(\fstate)$ captures the uncontrolled dynamics.
Our aim is to build reduced-order models on a slow SSM and determine an optimal projector that maps trajectories near the SSM onto trajectories on it, ensuring that off-manifold trajectories converge exponentially fast to their on-manifold counterparts. 
We demonstrate that these projections align with the system's stable fibers, as illustrated in Figure \ref{fig:foliation}.

To realize these fiber-aligned projections, we introduce a novel data-driven approach for learning optimal oblique projections onto the SSM. 
Our contributions are:
\begin{enumerate}
    \item A geometric characterization of neighborhoods around slow SSMs, demonstrating that projections aligned with stable fibers are optimal and can be approximated by oblique projections.
    \item A data-driven technique to learn these projections and the reduced dynamics.
    \item Numerical validation demonstrating how these projections improve prediction and control accuracy in slow-fast systems and high-dimensional robotics.
\end{enumerate}

\noindent 
\textbf{Outline}: %
Section~\ref{sec:related_work} reviews related work. Section~\ref{sec:structure} formally defines slow SSMs and their stable foliations.
Section~\ref{sec:projection_operator} establishes the optimality of fiber-aligned projections and introduces a data-driven approach to learning the projection operator.
Section~\ref{sec:numerical} applies SSMs to control and provides numerical validations and comparisons, with conclusions in Section~\ref{sec:conclusion}.
For brevity, we refer the reader to proofs in \iftoggle{ext}{\Cref{app:proof_flow_estimates,app:proof_foliation_rate,app:proof_wmap_lipschitz,app:proof_projection_operator}.}{the appendix of the extended version: {\small \url{https://arxiv.org/abs/2504.03157v2}.}}

\section{Related Work}\label{sec:related_work}
Linear projection-based methods, such as POD, often utilize oblique projections (\ie Petrov-Galerkin) to address the challenges of reducing systems which exhibit significant nonnormality, and transient growth (see \cite{benner_reviewPOD_2015} for a comprehensive review). More rigorous alternatives like SSM-based model reduction also utilize oblique projections, but unlike the data-driven projections in POD, these are inferred directly from a modal analysis of the slow-fast splitting of the governing equations \cite{jain_2022_SSMTool}.

While data-driven SSMs have shown state-of-the-art performance in various applications spanning fluids and soft robot control \cite{cenedese2022data-driven, alora2025discovering, kaundinya2025data-driven} the fast and slow subspaces of the systems in these applications are often non-orthogonal at the fixed point from which the SSM emanates \cite{bettini_oblique_2025}. To tackle this for 2D SSMs, Bettini et al.~\cite{bettini_oblique_2025} optimize an oblique projection using backbone data from mechanical vibration experiments. This local approximation of the fast stable invariant foliation is a more efficient alternative to learning global, nonlinear fiber approximations \cite{szalai2023_ISF}, since the latter demands significant amounts of data.

Alternatively, autoencoder-based reduction strategies project dynamics onto slow manifolds, determining their dimension by minimizing training error \cite{lee_2020_autoencoder, Champion_2019_autoencoder, fresca_2021_autoencoder}. These methods, however, often project orthogonally onto the target structure, introducing errors. Otto et al.~\cite{otto2023learning} refine this by incorporating nonlinear oblique projections, learning an optimal encoder to minimize the gap between projected and reduced dynamics. While effective, this approach requires knowledge of the governing equations and is data-intensive, also rendering it impractical for hardware implementations.

Inspired by the effectiveness of optimal linear oblique projections in experimental settings~\cite{bettini_oblique_2025} and the broad generality of nonlinear projections~\cite{otto2023learning}, we propose a method that learns an optimal linear oblique projection onto SSMs directly from data. This approach strikes a valuable trade-off, blending the practicality of \cite{bettini_oblique_2025} with the versatility of \cite{otto2023learning}, making it well-suited for diverse applications, including high-dimensional systems like soft robots. We further extend this method to nonlinear systems using more recent results on Normally Attracting Invariant Manifolds (NAIMs)~\cite{eldering2013normally, eldering2018global}.

\section{Structure of a Slow SSM}\label{sec:structure}
Our analysis of slow SSMs relies on two assumptions on the nonlinear system~\eqref{eq:ODE_intro}.

\begin{assumption}\label{assumption:f}
$\fode$ is continuously differentiable ($\fode \in \scalebox{0.95}{$C^1$}$), globally Lipschitz, and its Jacobian \scalebox{0.95}{$\jac \fode$} is Lipschitz.
\end{assumption}

\begin{assumption}\label{assumption:stability}
The uncontrolled system, $\fode$, has a \textit{single}\footnote{The analysis extends to systems with multiple equilibria, but remains local around a chosen equilibrium.} stable equilibrium point at the origin, \ie $\fode(\zero) = \zero$ and all eigenvalues of the Jacobian $\mathbf{A} \defeq \jac \fode(\zero)$ have negative real parts.
\end{assumption}

Assumption \ref{assumption:f} is a standard smoothness assumption \cite{Lorenz2005,Cannarsa2006} that guarantees the existence and uniqueness of solutions to the ODEs in \eqref{eq:ODE_intro}. By multiplying $\fode$ with a smooth cutoff function whose arbitrarily large support contains states of interest, the Lipschitz continuity assumptions are always satisfied if $\fode \in C^2$.

Assumption \ref{assumption:stability} is well-motivated as most engineering systems are stable at some equilibrium configuration. Although these systems naturally return to equilibrium when unforced, we require control to make them follow desired trajectories away from this stable point.

\subsection{Slow SSM}
The dominant dynamics of \Cref{eq:ODE_intro} are those that capture long-term behavior of the system. To study these persistent dynamics, we invoke the concept of a slow SSM. Broadly speaking, slow SSMs are low-dimensional attracting invariant manifolds that are tangent to slow eigenspaces at a stable hyperbolic fixed point. As such, SSMs capture the slow dynamics of the system, acting as the stable core onto which the fast, transient dynamics synchronize. Below, we expand on this in more detail.

We define an $\nrstate$-dimensional spectral subspace $\spectralsubspace$ as the direct sum of an arbitrary collection of $\nrstate$ real eigenspaces of $\mathbf{A}$:
\begin{align}
    E \defeq  E_{j_1} \oplus E_{j_2} \oplus ... \oplus E_{j_\nrstate},
\end{align}
where each $\spectralsubspace_{j_k}$ is the real eigenspace associated with the eigenvalue $\lambda_{j_k}$ of $\mathbf{A}$. Let $\Lambda_{\spectralsubspace}$ be the set of eigenvalues corresponding to $\spectralsubspace$, and $\Lambda_{\text{out}}$ be the complementary set of eigenvalues. If the spectral gap condition $$\min_{\lambda\in\Lambda_{\spectralsubspace}} \text{Re}(\lambda)>\max_{\lambda\in\Lambda_{\text{out}}} \text{Re}(\lambda)$$
holds, then $\spectralsubspace$ represents the slowest spectral subspace of order $\nrstate$.
Intuitively, this subspace captures the dominant modes responsible for the persisting dynamics of the system.
\begin{definition}[Slow SSM] \label{defn:ssm}
Let $\flow:\R^{\nfstate}\times\R^+\to\R^{\nfstate}$ denote the flow generated by the uncontrolled dynamics $\fode$. A manifold $\manifold \subset \R^{\nfstate}$ is a slow Spectral Submanifold (SSM) associated with the spectral subspace $\spectralsubspace$ if it satisfies the following conditions:
\begin{enumerate}[label=(\roman*)]
    \item \textbf{Invariance}: $\manifold$ is invariant under the flow $\flow(\cdot)$ \ie $\flow(\manifold) = \manifold$, for all $t \geq 0$.
    \item \textbf{No Center Manifolds}: At each point $\mathbf{p}\in \mathcal{M}$, there is a direct-sum decomposition
\[
T_{\mathbf{p}} \bigl(\mathbb{R}^{n_f}\bigr) \;=\; T_{\mathbf{p}}(\mathcal{M}) \;\oplus\; \E_{\pp}^s,
\]
where $\E_{\pp}^s$ is the $(n_f-n)$-dimensional subspace of transverse directions that converge to $\mathcal{M}$ and $T_{\mathbf{p}}$ denotes the tangent space at $\mathbf{p}$.
    \item \textbf{Uniform Attraction}: There exist constants $C>0$, and rates $\lambda>\mu>0$, such that for all $\pp \in \manifold$ and $t > 0$:
    \begin{equation}\label{eq:uniform_attraction}
    \begin{aligned}
        \norm{\jac \flow(\p) \bm{u}} &\leq C e^{-\lambda t} \norm{\bm{u}}, \> \bm{u} \in \E_\pp^s, \\
        \norm{\jac \flow(\pp)\bm{v}} &\leq C e^{-\mu t} \norm{\bm{v}}, \> \bm{v} \in T \manifold.
    \end{aligned}
    \end{equation}
    \item \textbf{Tangency}: $\manifold$ is tangent to $E$ at the origin and has the same dimension as $E$.
    \item \textbf{Unique Smoothest}: $\manifold$ is strictly smoother than any other invariant manifold.
\end{enumerate}
\end{definition}
This definition is a restatement of the slow SSM definition in \cite{haller2016nonlinear}, combined with the fact that a slow stable SSM with a single fixed point is a uniform NAIM \cite{eldering2013normally, hirsch1970invariant}.
The low-dimensionality of 
$\manifold$ manifests due to the uniform attraction condition \eqref{eq:uniform_attraction}, which ensures that the contraction and expansion rates normal to 
$\manifold$ dominate those within it. As a result, the long-term dynamics of 
\Cref{eq:ODE_intro} are effectively confined to 
$\manifold$, making it a natural candidate for reduced-order modeling.

\subsection{Local Coordinates and Flow Near $\manifold$}
To study the dynamics near a \textit{compact} slow spectral submanifold $\manifold$, we construct local coordinate systems using standard manifold theory \cite{wiggins1994normally, fenichel1974asymptotic}.  Because our analysis is local and focused on behavior close to 
$\manifold$, we now introduce a suitable neighborhood around it. Specifically, we define the tubular neighborhood of $\manifold$, denoted by $\tubular$, by first considering the normal $\varepsilon$-neighborhood
\begin{align} \label{eq:bundledefn}
    \normalbundleeps = \left\{ (\pp, \vv) : \pp \in \manifold,  \vv \in \E^s_\pp,  \norm{\vv} < \varepsilon \right\}
\end{align}
and applying the map $\mathbf{h}: \mathbf{N}^{s,\varepsilon} \rightarrow \mathbb{R}^{n_f},$ given by $\mathbf{h}(\pp, \vv) = \pp + \vv.$
Thus, we have $\tubular = \mathbf{h}(\normalbundleeps).$ 
We can think of this as an inflated region around $\manifold$ which extends outward in directions transverse to its surface.

To establish the local coordinate systems along $\manifold$, we cover it with a finite collection of open sets ${U_i}$. Each $U_i$ is equipped with:
\begin{enumerate}[label=(\roman*)]
    \item a coordinate chart $\chart_i: U_i \rightarrow \R^\nrstate $ that assigns reduced coordinates $\rstate = \chart_i(\mathbf{\pp})$ to each point $\pp \in U_i \subset \manifold$,
    \item a normal coordinate map $\normalchart_i: \normalbundleeps\mid_{U_i} \rightarrow \R^{\nfstate - \nrstate}$, assigning transverse coordinates $\nstate = \normalchart_i(\pp, \vv)$.
\end{enumerate}
With an appropriate atlas of $\manifold$, \ie $\bigcup_{i}(U_i, \chart_i)$, then, for a sufficiently small $\varepsilon$ in Equation \eqref{eq:bundledefn}, $\mathbf{h}$ is guaranteed to be a $C^{r-1}$-diffeomorphism.

The combined coordinate map $ (\chart_i \times \normalchart_i)(\pp, \vv) \defeq (\chart_i(\pp), \normalchart_i(\pp, \vv))$, where $(\chart_i \times \normalchart_i): \normalbundleeps \mid_{U_i} \rightarrow \R^\nrstate \times \R^{\nfstate - \nrstate}$, provides a local representation near $\manifold$. For clarity, we omit the index $i$ and denote these maps as $\chart$ and $\normalchart$.
We now define a global change of basis map across the different local coordinate systems:
\begin{equation}
    \wmap(\rstate, \nstate) \defeq \mathbf{h} \circ (\chart \times \normalchart)^{-1}(\rstate, \nstate).
\end{equation}
Intuitively, the map $\wmap$ tells us how to reconstruct a point in the full ambient space (of dimension~$\nfstate$) given its local reduced coordinates. This map is a $C^{r-1}$-diffeomorphism by construction and hence has inverse $\wmap^{-1}$.\footnote{\label{fn:changebasissingle}For example, $\fstate = \wmap(\rstate, \nstate) = \mathbf{W} \begin{bmatrix}
    \rstate^\top & \nstate^\top
\end{bmatrix}^\top$ where $\mathbf{W} \in \R^{\nfstate \times \nfstate}$ is a change of basis matrix for a single coordinate system.}
In the following, we show that this change of basis function is bounded above and below.
\begin{lemma} \label{lem:wmap_lipschitz}
    Suppose $\fstate$ and $\fstate'$ are in the tubular neighborhood of $\manifold$, \ie $\fstate, \fstate' \in \tubular$, then $\wmap$ is a bi-Lipschitz map which satisfies
\begin{equation}
\small
\begin{aligned}
\sigma_{\text{min}}\big(\jac \wmap\big)\norm{
\begin{bmatrix}
    \rstate - \rstate' \\
    \nstate - \nstate'
\end{bmatrix}
}
&\leq \norm{\wmap(\rstate, \nstate) - \wmap(\rstate', \nstate')} \\
&\leq \sigma_{\text{max}}\big(\jac \wmap\big)\norm{
\begin{bmatrix}
    \rstate - \rstate' \\
    \nstate - \nstate'
\end{bmatrix}
} ,
\end{aligned}
\end{equation}
where $\sigma_{\text{min}}\big(\jac \wmap\big)$ and $\sigma_{\text{max}}\big(\jac \wmap\big)$ are the minimum and maximum singular values of the Jacobian of $\wmap$, respectively.
\end{lemma}
\begin{proof}
    \iftoggle{ext}{See Appendix \ref{app:proof_wmap_lipschitz}.}{See appendix in the extended version.}
\end{proof}
In the single coordinate system case\footref{fn:changebasissingle}, the map is simply bounded by the singular values of the change of basis matrix.

For points close enough to the manifold, we can give local expressions of $\flow$ in the coordinates defined above:
\begin{equation} \label{eq:localflows}
\begin{aligned}
    \varphi_t(\rstate, \nstate) \defeq \chart \circ \mathbf{h}^{-1} \circ \flow \circ \wmap (\rstate, \nstate),
\end{aligned}
\end{equation}
where $\mathbf{h}^{-1}$ returns a tuple $(\pp, \vv)$ and $\chart$ is applied only to $\pp$.
Hence, $\varphi_t(\rstate,\nstate)$ represents how the reduced coordinates evolve under the flow.
        
\subsection{Foliations of a Slow SSM}
In this section, we discuss the basin of attraction of $\manifold$. It turns out that $\manifold$ is buttressed by stable manifolds which admit an internal structure called an \textit{invariant foliation} dictating the asymptotic behavior of trajectories onto $\manifold$. To describe this structure, let us consider the family of maps of the form:
\begin{align} \label{defn:fiber}
    \fiber(\cdot; \pp) : \normalchart(\pp, \E_\pp^s) \to \R^\nrstate, \quad  \fiber(\zero; \pp) = \chart(\pp).
\end{align}
These maps provide a systematic way of assigning reduced coordinates to points located transversely to $ \manifold $ along the stable directions. 

Using these maps, we introduce the \emph{stable fiber} associated with each base point $ \pp \in \manifold $ in local coordinates:
\begin{align}
    \mathcal{F}(\pp) \defeq \big\{ \big( \fiber(\nstate; \pp \big), \nstate) \> : \> \nstate \in \normalchart(\pp, \E_\pp^{s}) \big\}.
\end{align}

We say that $\qq$ is on the stable fiber of the base point $\pp$ when $\qq \in \mathcal{F}(\pp)$. The stable foliation, $\foliation{\manifold}$, is then given by
\begin{align}
    \foliation{\manifold} = \bigcup_{\pp \in \manifold} \mathcal{F}(\pp).
\end{align}
We now state the invariant foliation theorem \cite{wiggins1994normally,fenichel1974asymptotic} which outlines the properties of the stable structure around $\manifold$.
\begin{theorem}[Foliation of Stable Manifolds] \label{thm:foliations}
    Let $\manifold$ be a compact $C^r$-smooth slow SSM of system~\eqref{eq:ODE_intro}, such that $r \in \mathbb{N}^{+} \cup\{\infty, a\}$ where $a$ refers to analytic. Then there exists an $\nfstate$-parameter family of $(\nfstate-\nrstate)$-dimensional stable fibers $\foliation{\manifold}$ satisfying:
    \begin{enumerate}[label=(\roman*)]
        \item Each fiber $\mathcal{F}(\pp)$ is an $(\nfstate - \nrstate)$-dimensional, $C^\nsmooth$-smooth manifold.
        \item $\mathcal{F}(\pp)$ is tangent to $\E^s_\pp$ at $\pp$.
        \item \label{enum:limit} Suppose $\qq \in \mathcal{F}(\pp)$ and $\qq' \in \mathcal{F}(\pp')$ where $\pp \neq \pp'$, then
        \begin{equation} \label{eq:rate_foliation}
            \lim_{t \to \infty} \frac{\norm{\flow(\qq) - \flow(\pp)}}{\norm{\flow(\qq') - \flow(\pp)}} = 0 .
        \end{equation}
        \item $\mathcal{F}(\pp) \cap \mathcal{F}(\pp') = \emptyset$, unless $\pp = \pp'$.
        \item $\mathcal{F}(\pp)$ are $C^{r-1}$ with respect to its base point $\pp$.
    \end{enumerate}
\end{theorem}
\begin{proof}
    By definition, a slow SSM is a NAIM \cite{eldering2013normally}. We then invoke results from Theorem 1 in \cite{eldering2018global} that guarantee the existence of a stable foliation. Since a NAIM is a special case of a normally hyperbolic invariant manifold, the properties of the stable foliation as stated in \cite{wiggins1994normally} apply.
\end{proof}
\Cref{thm:foliations} states that the stable foliation has dimension equal to the co-dimension of $\manifold$ and that each fiber is tangent to the fast subspace. Furthermore, it states that trajectories starting on the same fiber converge to each other at the fastest rate, each fiber has a unique base point, and that the fibers vary smoothly along the manifold.
The geometry of the stable fiber is shown in Figure \ref{fig:foliation}. These stable fibers are the key to characterizing the optimal projection operator onto the slow SSM.

\section{Optimal Projection Operator}%
\label{sec:projection_operator} 
In this section, we show that the fibers implicitly characterize the optimal projection operator onto the slow SSM.
We first prove that projecting initial conditions onto their corresponding fiber's base point minimizes the integrated error.
A zeroth-order approximation of $f^s(\pp)$ allows us to derive a linear oblique projection onto $\manifold$, for which we propose a data-driven learning algorithm.

\subsection{Optimality of the Base Point} \label{sec:basepoint_optimality}
We now establish the main result demonstrating the optimality of fiber-aligned projections for initial conditions near the manifold. Specifically, we show in Theorem \ref{thm:optimal_ise} that, in an appropriate neighborhood, projecting any off-manifold initial condition $\qq$ onto its unique corresponding base point $\pp \in \manifold$ minimizes the long-term trajectory prediction error. The proof of this theorem requires intermediate convergence results provided by Corollary \ref{corr:asymptotic_rate_foliation} and explicit exponential convergence estimates from Lemma \ref{lemma:flow_estimates}.

Any off-manifold point $\qq$ in the neighborhood $\tubular$ belongs uniquely to a fiber $\mathcal{F}(\pp)$ anchored at a base point $\pp$. Trajectories initialized at $\qq$ converge faster to the trajectories starting from this base point $\pp$ than to trajectories from any other manifold point $\pp' \neq \pp$, as stated formally below:

\begin{corollary}
\label{corr:asymptotic_rate_foliation}
    Suppose $\qq \in \mathcal{F}(\pp)$ where $\pp \neq \pp'$ are distinct points on $\manifold$, then
        \begin{equation} \label{eq:basepoint_convergence}
            \lim_{t \to \infty} \frac{\norm{\flow(\qq) - \flow(\pp)}}{\norm{\flow(\qq) - \flow(\pp')}} = 0.
        \end{equation}
\end{corollary}
\begin{proof}
    This follows from Theorem \ref{thm:foliations} \ref{enum:limit}; for  the proof, \iftoggle{ext}{see Appendix \ref{app:proof_foliation_rate}.}{see appendix in the extended version.}
\end{proof}
To guarantee that the convergence in~\eqref{eq:basepoint_convergence} is uniform, we introduce the following assumption\footnote{This assumption can, in principle, be rigorously justified by invoking continuity of the flow and compactness of the chosen neighborhood. We omit the detailed argument here for simplicity.}:
\begin{assumption}\label{assum:closeness} There exists a neighborhood ${\pp' : \norm{\pp - \pp'} < \gamma_1}$ of the base point $\pp$, with $\gamma_1 > 0$, such that for every $\epsilon > 0$, there is a finite time $T_1$, independent of $\pp'$, satisfying \begin{align} \label{eq:finite_asymptotic_estimate} \norm{\flow(\qq) - \flow(\pp)} < \epsilon \norm{\flow(\qq) - \flow(\pp')}, \end{align} for all $\pp'$ within this neighborhood and all $t > T_1$. \end{assumption}

The following lemma provides bounds on the convergence rate of trajectories from nearby manifold points.
\begin{lemma} \label{lemma:flow_estimates}
    Let \(\qq \in \mathcal{F}(\pp)\) be an off-manifold initial condition with corresponding base point \(\pp\). Then, there exists a constant \(\gamma_2 > 0\) such that for any \(\pp'\) satisfying \(\norm{\pp - \pp'} < \gamma_2\), the following exponential bound holds:
    \begin{equation}
    \norm{\varphi_t(\rstatebar, \zero) - \varphi_t(\rstate', \zero)} \leq K e^{-\mu t}\norm{\rstatebar - \rstate'},
    \end{equation}
where $\rstatebar = \chart(\pp), \rstate' = \chart(\pp')$, and \(\mu > 0\) and \(K > 0\) are positive constants independent of \(t\).
\end{lemma}
\begin{proof}
    \iftoggle{ext}{See Appendix \ref{app:proof_flow_estimates}.}{See appendix in the extended version.}
\end{proof}

We now show that a fiberwise projection to the appropiate base point is optimal in the sense that it minimizes long-term integrated error. 

\begin{theorem} \label{thm:optimal_ise}
Consider a nonlinear dynamical system defined by \eqref{eq:ODE_intro} with flow $\flow$ near the slow SSM $\manifold$. Let $\qq \in \mathcal{F}(\pp) \setminus \manifold$ be an off-manifold initial condition with corresponding base point $\pp \in \manifold$. Then, there exists a neighborhood $\delta > 0$ around $\pp$ such that for all $\pp' \in \mathcal{M}$ satisfying $\|\pp'-\pp\| < \delta$, the base point $\pp$ is optimal in the sense that
\begin{equation} \label{eq:projection_minimizer}
    \pp = \argmin_{\pp'} \int_0^\infty \|\flow(\qq) - \flow(\pp')\| \, dt.
\end{equation}
\end{theorem}

\begin{proof}
Our goal is to verify that for any $\pp' \in \mathcal{M}, \; \pp' \neq \pp$ satisfying $\|\pp'-\pp\| < \delta$ we have
\begin{equation} \label{eq:maininequality}
\small
\begin{aligned}
\int_0^\infty \norm{\flow(\qq) - \flow(\pp)}dt  < \int_0^\infty \norm{\flow(\qq) - \flow(\pp')}dt.
\end{aligned}
\end{equation}

Fix $\epsilon \in (0,1)$ to obtain a $T_1 > 0$ from \eqref{eq:finite_asymptotic_estimate}. Note there exists a $\gamma_3>0$ ensuring
\begin{align*}
    M = \inf_{\pp' \in \{\hat{\pp} : \norm{\pp - \hat{\pp}} < \gamma_3 \}} \int_{T_1}^\infty \norm{ \flow(\qq) - \flow(\pp')} dt > 0.
\end{align*}
Choose $\delta < \min(\gamma_1, \gamma_2, \gamma_3)$ so that the above, \Cref{assum:closeness}, and \Cref{lemma:flow_estimates} all hold.

Integrating Equation \eqref{eq:finite_asymptotic_estimate} and splitting the integrals at $T_1$, we find that it suffices to show
\begin{equation} \label{eq:final_inequality}
\small
\begin{aligned}
    &\int_{0}^{T_1} \hspace{-0.5em}\Big(\norm{\flow(\qq) - \flow(\pp)} -  \norm{\flow(\qq) - \flow(\pp')}\Big)dt \\
    &\leq \int_{0}^{T_1} \hspace{-0.5em}\Big(\norm{\flow(\pp) - \flow(\pp')} \Big)dt \leq (1-\epsilon)M.
\end{aligned}
\end{equation}

Since $\pp, \pp' \in \manifold$, we have that $(\rstatebar, \zero) = (\chart \times \normalchart)(\pp)$ and $(\rstate', \zero) = (\chart \times \normalchart)(\pp')$. Thus, we can expand
\begin{equation}
\small
\begin{aligned}
    \norm{\flow(\pp) - \flow(\pp')} &= \norm{\wmap\big(\varphi_t(\rstatebar, \zero), \zero \big) - \wmap\big(\varphi_t(\rstate', \zero), \zero \big)} \\
    & \leq \sigma_{\text{max}}(\jac \wmap) \norm{\varphi_t(\rstatebar, \zero) - \varphi_t(\rstate', \zero)} \\
    & \leq \sigma_{\text{max}}(\jac \wmap) K e^{-\mu t}\norm{\rstatebar - \rstate'},
\end{aligned}
\end{equation}
where the first inequality is due to \Cref{lem:wmap_lipschitz} and the second is due to \Cref{lemma:flow_estimates}. Plugging this into \eqref{eq:final_inequality} and integrating over $[0,T_1]$ gives
\begin{equation}
\small
\begin{aligned}
     \frac{\sigma_{\text{max}}(\jac \wmap) K}{\mu}\left(1 - e^{-\mu T_1}\right) \norm{\rstatebar - \rstate'} \leq (1-\epsilon)M. 
\end{aligned}
\end{equation}
Since $\norm{\rstatebar - \rstate'} < \delta$, we can choose \(\delta\) small enough so that
\begin{align}
    \delta \leq \frac{\mu(1-\epsilon)M}{\sigma_{\text{max}}(\jac \wmap)K(1 - e^{-\mu T_1})},
\end{align}
ensures the desired inequality.
\end{proof}

The time $T_1$ in \eqref{eq:finite_asymptotic_estimate} represents the time it takes for transient dynamics to synchronize with the dynamics on the manifold. If this synchronization occurs rapidly ($T_1 \to 0$), the radius $\delta$ grows arbitrarily large, effectively making the projection optimal for the entire manifold. Conversely, if convergence to the manifold is slower (larger $T_1$), the neighborhood size $\delta$ shrinks accordingly. This occurs because the prolonged transient dynamics increase the finite-time error between trajectories starting from 
$\pp$ and nearby points $\pp'$, necessitating a smaller neighborhood to ensure the projection remains optimal. Essentially, restricting the neighborhood of optimality around $\pp$ ensures that the initial transient error within the integral never outweighs the long-term optimality of the projection.

\subsection{Form of Projection Operator} \label{sec:form_of_operator}
We aim to learn the chart and parameterization maps of $\manifold$ via
\begin{equation}
\begin{aligned}
    \vmap(\qq) &\defeq \chart \circ \mathbf{h}^{-1} (\qq), \\
    \parammap(\rstate) &\defeq \wmap(\rstate, \zero),
\end{aligned} \label{eq:chart_param}
\end{equation}
which in turn defines the following fiberwise projection map
\begin{equation} \label{eq:fiberwiseprojection}
    \Pi(\qq) \defeq \parammap \circ \vmap(\qq),
\end{equation}
such that $\Pi : \tubular \to \manifold$ assigns to each point $\qq$ the unique base point $\pp$ of the stable fiber, $\mathcal{F}(\pp)$.
The following proposition provides a local representation of the stable fibers as a graph over the fast transverse directions, $\mathbf{E}^s_\pp$ and shows an approximation of the projection.
\begin{proposition} \label{prop:projection}
    For a slow SSM, $\manifold \subset \R^{\nfstate}$ and base point $\pp = (\rstatebar, \zero)$, the stable fiber $\mathcal{F}(\pp)$ can be approximated  as
    \begin{equation}
    \begin{aligned}
        \mathcal{F}(\pp) = \Bigg\{ 
        \begin{bmatrix}
            \rstatebar + \Vconst \nstate \\ \nstate
        \end{bmatrix}
        \> : \> \nstate \in \R^{\nfstate - \nrstate}
        \Bigg\} ,
    \end{aligned}
    \end{equation}
    where $\Vconst = \jac_{\nstate} \fiber(\zero; \zero)$ is a constant $\R^{\nrstate \times (\nfstate - \nrstate)}$. The candidate projection operator $\Pi : \R^{\nfstate} \to \manifold$, mapping $\qq \in \mathcal{F}(\pp)$ to $\pp$, is
    \begin{equation} \label{eq:projection}
        \Pi = \Vpca \Vopt^\top,
    \end{equation}
    where $\Vpca = \begin{bmatrix}
        \eye_{\nrstate \times \nrstate} & \zero_{\nrstate \times (\nfstate - \nrstate)}
    \end{bmatrix}^\top$ and $\Vopt = \begin{bmatrix}
        \eye_{\nrstate \times \nrstate} & -\Vconst
    \end{bmatrix}^\top$.
\end{proposition}
\begin{proof}
    The proof Taylor expands the fiber map and is detailed in 
    \iftoggle{ext}{Appendix \ref{app:proof_projection_operator}.}{the appendix of the extended version.}
\end{proof}
Intuitively, Proposition \ref{prop:projection} shows that nearby fibers point approximately in the same direction as the stable fiber at the equilibrium point.
Since, in practice, we work in ambient coordinates, $\Vpca$ and $\Vopt$ may take a different form and must be identified. The procedure for inferring these matrices from data is detailed in \Cref{sec:learn_obique}.
%

We now show that $\Pi$ with an appropriate $\Vconst$ is a proper projection.
\begin{theorem} \label{thm:projection}
    Suppose $\R^\nfstate = E \bigoplus \E^s_0$ is a direct-sum decomposition of an $\nrstate$-dimensional slow spectral subspace, $E$, and an $(\nfstate-\nrstate)$-dimensional fast subspace $\E^s_{\zero}$. Then $\Pi$ is a projection in the sense that it satisfies: (i) $\Pi^2 = \Pi$, (ii)  $\text{range}(\Pi) = E$, and (iii) $\text{ker}(\Pi) = \E^s_\zero$,
    where $\Vpca \in \R^{\R^{\nfstate \times \nrstate}}$ is a matrix whose columns span $E$ and $\Vopt^\top \in \R^{\nrstate \times \nfstate}$ is a map that annihilates elements in $\E^s_\zero$.
\end{theorem}
\begin{proof}
    From the definition of $\Pi$ \eqref{eq:projection}, we see that (i) is satisfied since $\Vopt^\top \Vpca = \eye_{\nrstate \times \nrstate}$ and (ii) is satisfied by construction. To show (iii), notice that $(\Vconst \nstate, \nstate)$ is annihilated by $\Pi$ for any $\nstate$. Thus, we must show that $(\Vconst \nstate, \nstate) \in \E^s_\zero$ for appropriately chosen $\Vconst$.

    Let us define $\mathbf{A} = \jac \fode(\zero)$. In a suitable basis, the matrix $\mathbf{A}$ can be arranged into a block-diagonal form
    \begin{equation}
        \mathbf{A} =
        \begin{bmatrix}
            \mathbf{A}_{TT} & \mathbf{A}_{TN}\\
            \zero & \mathbf{A}_{NN}
        \end{bmatrix},
    \end{equation}
    where $\mathbf{A}_{TT} \in \R^{\nrstate \times \nrstate}$ acts on the tangent directions, $\mathbf{A}_{TN} \in \R^{\nrstate \times (\nfstate - \nrstate)}$ accounts for first-order coupling from normal to tangent direction, and  $\mathbf{A}_{NN} \in \R^{(\nfstate - \nrstate) \times (\nfstate - \nrstate)}$ acts on the normal direction. To show that $(\Vconst \nstate, \nstate) \in \E^s_\zero$, we need to show that the subspace
    \begin{align*}
        U = \{(\rstate, \nstate) : \rstate = \Vconst \nstate \}
    \end{align*}
    is $\mathbf{A}$-invariant, \ie $\mathbf{A}(U) \subset U$. For $U$ to be invariant under $\mathbf{A}$, there must exist some $\nstate'$ such that
    \begin{align}
        \begin{bmatrix}
            \mathbf{A}_{TT} \Vconst \nstate + \mathbf{A}_{TN} \nstate \\
            \mathbf{A}_{NN} \nstate
        \end{bmatrix}
        =
        \begin{bmatrix}
            \Vconst \nstate' \\
            \nstate'
        \end{bmatrix}
    \end{align}
    which implies we must find $\Vconst$ which solves the Sylvester equation
    \begin{align} \label{eq:sylvester}
        \mathbf{A}_{TT} \Vconst - \Vconst \mathbf{A}_{NN} = -\mathbf{A}_{TN}.
    \end{align}
    There exists a unique solution to \eqref{eq:sylvester} iff $\text{spect}(\mathbf{A}_{TT}) \cap \text{spect}(\mathbf{A}_{NN}) = \emptyset$ \cite{golub2013matrix}, which is guaranteed in our setting due to the spectral gap between $E$ and $\E^s_\zero$.
\end{proof}

The projection operator $\Pi$ is valid only near the fixed point, where its linear approximation captures transient behavior. In the following section we will learn this projection operator across the phase space of the system.

\subsection{Learning Oblique Projections from Data}\label{sec:learn_obique}
\begin{algorithm}[t]
\small
\caption{Data Curation and Preparation} \label{alg:datacuration}
\begin{algorithmic}[1]
\Require \( \bm{\mathcal{Y}}, T_1 \)

\State Compute \(\dot{\bm{\mathcal{Y}}}\) using finite differencing
\State Partition \(\bm{\mathcal{Y}}, \dot{\bm{\mathcal{Y}}}\) at \(T_1\) into transient (\(\bm{\mathcal{Y}}_{\text{trans}}, \dot{\bm{\mathcal{Y}}}_{\text{trans}}\)) and near-manifold (\(\bm{\mathcal{Y}}_{\text{near}}, \dot{\bm{\mathcal{Y}}}_{\text{near}}\)) subsets
\State Perform SVD on \(\bm{\mathcal{Y}}_{\text{near}}\) to obtain \(\Vpca \in \mathbb{R}^{p \times \nrstate}\)

\State \textbf{Return} \(\boldsymbol{\mathcal{Y}}_{\text{trans}}, \dot{\bm{\mathcal{Y}}}_{\text{trans}}, \bm{\mathcal{Y}}_{\text{near}}, \dot{\bm{\mathcal{Y}}}_{\text{near}}, \Vpca\)
\end{algorithmic}
\end{algorithm}

\begin{algorithm}[t]
\caption{Learning the Oblique Projection}
\begin{algorithmic}[1] \label{alg:learnoblique}
\small
\Require \(\bm{\mathcal{Y}}_{\text{trans}}, \dot{\bm{\mathcal{Y}}}_{\text{trans}}, \Vpca, n_r\)
\State Co-optimize for \(\Vopt\) and \(\mathbf{R}_{\text{trans}}\):
\begin{equation}\label{eq:learn_trans}
\begin{aligned}
  \min_{\Vopt, \mathbf{R}_{\text{trans}}} & \quad \left\| \Vopt^{\top} \boldsymbol{\mathcal{\dot{Y}}}_{\text{trans}} - \mathbf{R}_{\text{trans}} \left( \Vopt^{\top} \boldsymbol{\mathcal{Y}}_{\text{trans}} \right)^{1:n_r} \right \|^2 \\
  \text{s.t.} & \quad \Vopt^{\top} \Vpca = \mathbf{I} ,
\end{aligned}
\end{equation}
\State Discard \(\mathbf{R}_{\text{trans}}\) as it reflects transient dynamics

\State \textbf{Return} \(\mathbf{V}_{\text{opt}}\)
\end{algorithmic}
\end{algorithm}

\begin{algorithm}[t]
\small
\caption{Learning Reduced Dynamics, Parameterization \\ Map, and Control Dynamics}
\begin{algorithmic}[1] \label{alg:learncontroldyn}
\Require \(\bm{\mathcal{Y}}_{\text{near}}, \dot{\bm{\mathcal{Y}}}_{\text{near}}, \bm{\mathcal{Y}}_u, \dot{\bm{\mathcal{Y}}}_u, \mathbf{U}, \Vopt, \Vpca, n_r, n_w\)
\State Solve the following optimization problems consecutively:
\begin{minipage}{\linewidth}
\begin{equation}
\begin{aligned} \label{eq:onmanifoldyns}
\min_{\mathbf{R}} &\left\| \Vopt^\top \dot{\bm{\mathcal{Y}}}_{\text{near}} - \mathbf{R} (\mathbf{V}_{\text{opt}}^\top \bm{\mathcal{Y}}_{\text{near}})^{1:n_r} \right\|^2 \\
\min_{\mathbf{W}_{\text{nl}}} &\left\| \bm{\mathcal{Y}}_{\text{near}} - \Vpca (\Vopt^\top \bm{\mathcal{Y}}_{\text{near}}) - \mathbf{W}_{\text{nl}} (\Vopt^\top \bm{\mathcal{Y}}_{\text{near}})^{2:n_w} \right\|^2 \\
\text{s.t.}& \quad \Vopt^\top \mathbf{W}_{\text{nl}} = \mathbf{0} \\
\min_{\mathbf{B}_r} &\left\| \Vopt^\top \dot{\bm{\mathcal{Y}}}_u - \mathbf{R} (\Vopt^\top \bm{\mathcal{Y}}_u)^{1:n_r} - \mathbf{B}_r \mathbf{U} \right\|^2
\end{aligned}
\end{equation}
\end{minipage}
\State \textbf{Return} \(\mathbf{R}, \mathbf{W}_{\text{nl}}, \mathbf{B}_r\)
\end{algorithmic}
\end{algorithm}

For real-world systems, we must learn the oblique projection, $\Pi$, from observational data.
Specifically, we fit the chart and parameterization maps as defined in \eqref{eq:chart_param} and the on-manifold dynamics, $\mathbf{r} : \R^n \to \R^n$.
The chart and parameterization maps must satisfy the invertibility relations \cite{haller2016nonlinear}, $\obs = (\parammap \circ \vmap) (\obs)$ and $\rstate = (\vmap \circ \parammap)(\rstate)$, where we replaced $\fstate$ with the observed states, $\obs \in \R^p$.
The exact form satisfying these relations we choose is
\begin{equation}
\begin{aligned}
    \rstate &= \vmap(\obs) \defeq \Vopt^\top \obs, \\
    \obs &= \parammap(\rstate) \defeq \Vpca \rstate + \mathbf{W}_{\text{nl}} \rstate^{2:n_w} ,
\end{aligned}\label{eq:chart_param_poly}
\end{equation}
since $\Vopt^\top \Vpca = \eye$ and we will enforce $\Vopt^\top \mathbf{W}_{\text{nl}} = \mathbf{0}$.
We use $\rstate^{2:n_w}$ to denote the family of all monomials from order 2 to $n_w$.
The dynamics on $\manifold$ takes a similar polynomial form
\begin{equation}\label{eq:reduced_ode}
    \dot{\mathbf{x}}_r = \mathbf{r} (\rstate) \defeq \mathbf{R} \rstate^{1:n_r} ,
\end{equation}
with $n_r$ being the desired order of the Taylor expansion.

To guarantee the embedding of an $n$-dimensional SSM, we require \(p \geq 2n+1\), which is commonly satisfied by enlarging the observable space via delay-embedded trajectories (see \cite{cenedese2022data-driven}). Differential embeddings are locally observable \cite{letellier2005relation}, which implies that delay embeddings are also likely locally observable (since the latter is a discrete version of the former). Thus, if the robot's sensors are appropriately placed to detect the dominant modes \cite{axas2023modelreduction} and the robot's controlled behaviors keep the system close to the SSM, the reduced state in \eqref{eq:reduced_ode} is approximately observable.

Observable data comprise of decaying trajectories, which we stack to form a dataset \(\bm{\mathcal{Y}} \in \mathbb{R}^{p \times N}\), where \(N\) is the trajectory length and $p$ is the dimension of the observable space.
We partition this dataset into transient and near-manifold components:
\(\bm{\mathcal{Y}}_{\text{trans}} \in \mathbb{R}^{p \times N_{T_1}}\) contains the first $N_{T_1}$ time steps corresponding to the transient dynamics, and \(\bm{\mathcal{Y}}_{\text{near}} \in \mathbb{R}^{p \times (N-N_{T_1})}\) captures the remaining near-manifold data.
Earlier methods, for example \cite{cenedese2022data-driven}, discarded \(\bm{\mathcal{Y}}_{\text{trans}}\), and computed the chart map, reduced  dynamics and SSM parametrization in the observable coordinates using \(\bm{\mathcal{Y}}_{\text{near}}\). \Cref{alg:datacuration} recalls these steps and obtains $\Vpca$, whose column vectors span the tangent space of the SSM.
Setting $\mathbf{V}_{\text{opt}} = \Vpca$ recovers an orthogonal projection onto the SSM.

From analysis in Section \ref{sec:basepoint_optimality}, we know \(\bm{\mathcal{Y}}_{\text{trans}}\) contains information of local fiberwise directions $\mathbf{V}_\text{opt}$.
\Cref{alg:learnoblique} therefore optimizes jointly over the dynamics in \eqref{eq:reduced_ode} and the projection direction to find an optimal $\mathbf{V}_{\text{opt}}$. 
Note that the constraint in \eqref{eq:learn_trans} enforces that the projection operator is idempotent (recall \Cref{thm:projection}(i)).
We employ projected gradient descent or IPOPT \cite{wachter2006implementation} to solve this non-convex optimization problem. 
Finally, we use the optimized oblique projection $\mathbf{V}_{\text{opt}}$ and near-manifold data to learn the coefficient matrices of the reduced dynamics and the parameterization map, $\mathbf{R}$ and $\mathbf{W}_{\text{nl}}$, as detailed in \Cref{alg:learncontroldyn}.
Indeed, the invertibility relations are satisfied as a result of both the objective function and the constraint in the optimization of $\mathbf{W}_{\text{nl}}$.

It is important to note that learning $\Pi$ from data far from the fixed point minimizes average error by fitting prevalent transients. While recovering the true nonlinear fibers would exactly minimize \eqref{eq:projection_minimizer}, this typically requires prohibitive amounts of data. Our method therefore strikes a balance between accuracy and practicality.

\section{Data-Driven Control via SSMs}\label{sec:numerical}
In Section~\ref{sec:learn_obique}, we described a method to learn optimal oblique projections for SSM-reduced models using data from uncontrolled systems. To apply this to controlled dynamics, we augment the autonomous reduced dynamics with a linear control contribution, \ie with slight abuse of notation,
\begin{equation*}
    \dot{\mathbf{x}}_r =  \mathbf{r} (\rstate) + \mathbf{B}_r \ctrl .
\end{equation*}

We apply control inputs to excite the system and collect the resulting observables, yielding the stacked matrices $\mathbf{U} \in \R^{m \times N^\prime}$ and $\boldsymbol{\mathcal{Y}}_u \in \R^{p \times N^\prime}$.
Using the optimization in \Cref{alg:learncontroldyn}, we compute the control matrix $\mathbf{B}_r$. This calibrated model enables closed-loop control with Model Predictive Control (MPC). Unlike standard MPC, our approach uses oblique projections of controlled data with $\Vopt$. For details on the optimal control problem formulation, we refer readers to 
\iftoggle{ext}{\Cref{app:mpc_formulation}.}{the appendix in the extended version.}
We now demonstrate the effectiveness of our approach on a benchmark problem and closed-loop control of a continuum trunk robot.
\footnote{All necessary code to reproduce the results can be accessed at \url{https://github.com/StanfordASL/Opt-SSM}.}

\begin{figure}[t]
    \centering
    \includegraphics[height=1.65in]{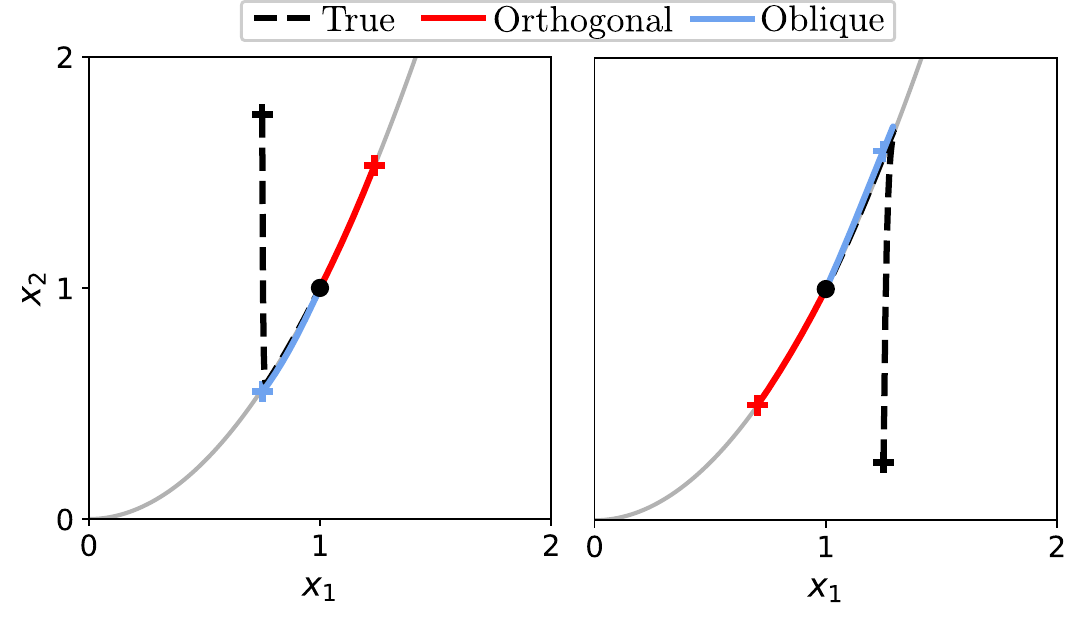}
    \caption{Comparison of predicted trajectories for SSMs with orthogonal projections and optimal oblique projections. An approximation of the critical manifold is shown in gray. The trajectories originate from \fplus ~and, after projection, converge to the stable fixed point at (1,1).}\label{fig:slow_fast_comparison}
\end{figure}

\subsection{Slow-Fast System}
To illustrate the importance of projection direction, we consider a two-dimensional system with large timescale separation, adopted from \cite{otto2023learning}.
We add forcing, $\mathbf{u} = \begin{bmatrix} u_1 & u_2 \end{bmatrix}^\top$, such that the governing equations become
\begin{equation}\label{eq:simple_fast_slow_forced}
\begin{aligned}
\dot{x}_1 & = \lambda x_1\left(1-x_1^2\right) + \alpha u_1 , \\
\varepsilon \dot{x}_2 & = x_1^2-x_2 + \varepsilon \beta u_2 ,
\end{aligned}
\end{equation}
where $\lambda, \varepsilon>0$ and $\varepsilon^{-1} \gg \lambda$.
We set the parameters as $\lambda = \varepsilon = 0.1$ and $\alpha = \beta = 0.2$.
For the uncontrolled system, we collect 10 trajectories around the asymptotically stable fixed point $(1,1)$, and shift the data to be centered around the origin. In this example, we observe the full state space and reduce to a 1-dimensional system.

The chart map and reduced dynamics are found using \Cref{alg:learnoblique} and the parameterization map using \Cref{alg:learncontroldyn}.
In Figure \ref{fig:slow_fast_comparison} (left), we show how SSMs using orthogonal projections compare with those using our optimized oblique projections.
We initialize off the slow manifold and, for both methods, project onto it before propagating the reduced dynamics.
The orthogonal projection causes the predicted response to be on opposite side of the stable fixed point, resulting in large trajectory error.
The oblique projection, on the other hand, projects vertically along the fast dynamics, producing minimal trajectory error.
Additionally, we simulate the forced response with $\mathbf{u} = \begin{bmatrix}e^{-t} & e^{-t}\end{bmatrix}^\top$.
The results in Figure \ref{fig:slow_fast_comparison} (right) show again improved prediction accuracy.

\begin{figure}[t]
    \centering
    \includegraphics[height=2.275in]{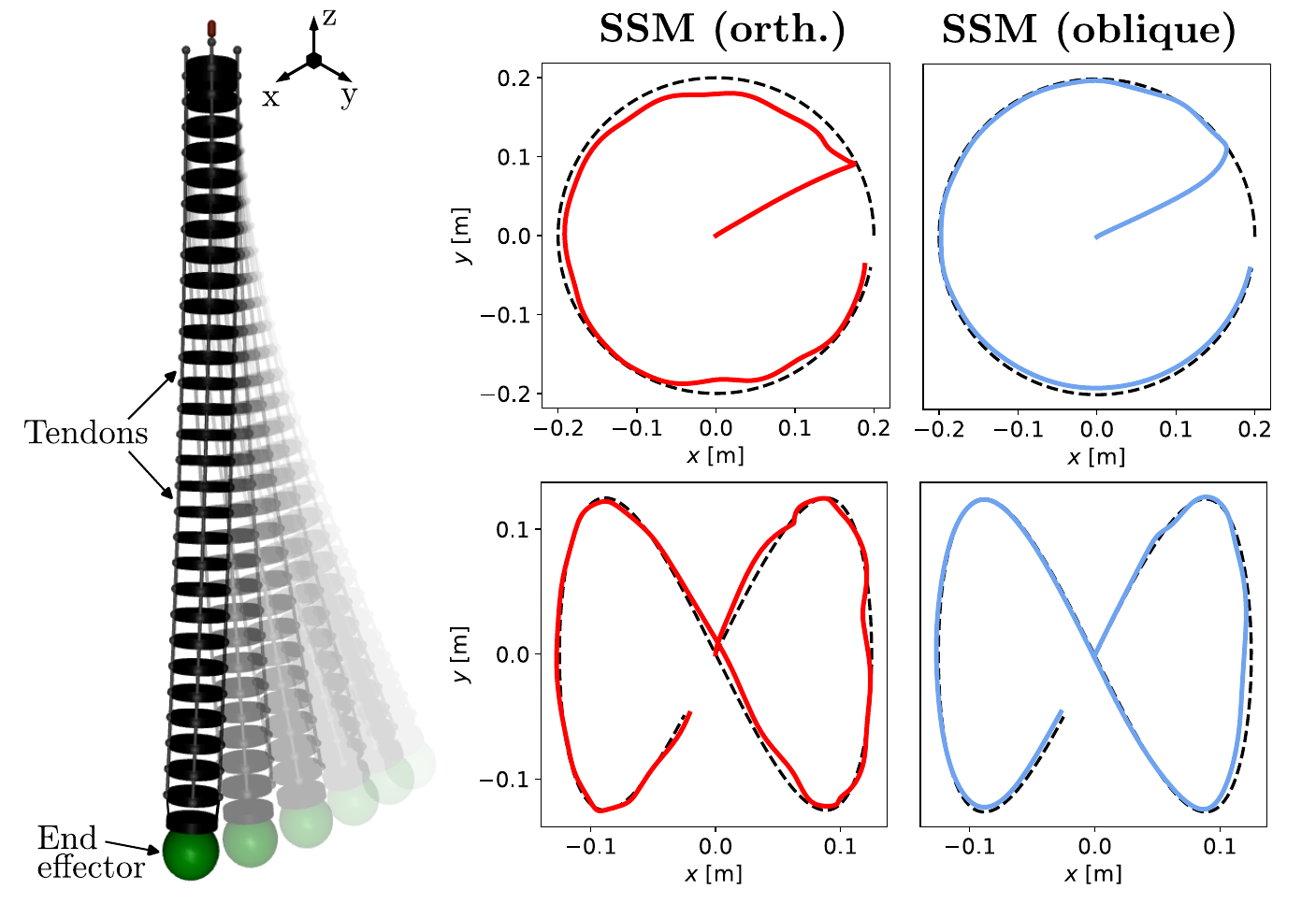}
    \caption{The simulated trunk robot (left) has a 30-link (180 dim.) body actuated by 12 antagonistic tendons ($\mathbf{u} \in \R^6$). We compare closed-loop tracking performance of the standard SSM scheme with our proposed oblique projection-based SSM.}\label{fig:trunk}
\end{figure}

\subsection{High-Dimensional Trunk Robot}
Finally, we showcase the effectiveness of our framework on a high-dimensional system.
We simulate the system, depicted in Figure \ref{fig:trunk}, using MuJoCo and perform closed-loop control using our MPC control strategy.

We collect 10 autonomous decay trajectories for fitting of the chart map, parameterization map and reduced dynamics.
Each trajectory is $10s$ in duration, and observations are stored at a time discretization of $dt=0.01s$.
We concatenate the observations with 3 time-delays and learn a 5D autonomous SSM model, with $n_r = n_w = 2$.
These parameter values align with the findings for similar SSM-reduced models in soft trunk robot geometries (see \cite{kaundinya2025data-driven} for the justification).
For the orthogonal projection, we report a mean MSE of $9.16 \pm 5.76$ $cm^2$ across 10 test decay trajectories, while the oblique projection achieves $1.60 \pm 1.10$ $cm^2$.
An additional 20 controlled trajectories are collected for regression of the control matrix.
Specifically, we excite the system by harmonically varying the control inputs.

Two trajectory tracking tasks are constructed, a circle of radius $20cm$ and a figure eight of radius $12.5cm$.
The circle is chosen such that it lies on the border of the training data.
Both reference trajectories have a duration of $2.5s$ and are defined in the horizontal plane.
We solve the MPC formulation with a horizon of $N=8$, and execute the first two control inputs of each horizon.
As a closed-loop performance metric we use the Integrated Square Error (ISE), and for both tasks compare the performance of SSMs using orthogonal and oblique projections in Table \ref{tab:trunk_results}.
These results show that our optimal projection outperforms orthogonal projection, which we attribute to the improved predictive accuracy.

\section{Conclusion and Outlook}\label{sec:conclusion}
In this work, we have introduced optimal oblique spectral projections to enhance the accuracy of data-driven model reduction to slow SSMs.
By utilizing the stable fiber structure surrounding slow SSMs, we developed a linear approximation of these fibers that reduces long-term prediction errors in the system's dynamics.
Additionally, we have shown how to effectively learn these projections from data and use them to achieve notable improvements in closed-loop control performance.

Future work could extend our framework by learning nonlinear chart maps, $\Vopt(\rstate)$, in which case we expect to approximate the curvature of the fibers.
Moreover, we assume user-specified transient segmentation, $T_1$, which may yield suboptimal results if misspecified. Recent advances in inferring dynamical structures from data \cite{mcinroe2024global} could help automate this partitioning.

\begin{table}[bt!]
    \centering
    \caption{Integrated Square Error (ISE) values (in $cm^2 s$) for both methods on two evaluation tasks.}
    \begin{tabular}{p{2.25cm}  p{1.75cm} p{1.75cm}}
        \toprule
        & SSM (orth.)  & SSM (oblique)  \\
        \midrule
         Circle (2D) & 4.70 & \textbf{0.94} \\ 
         Figure Eight (2D) & 0.87 & \textbf{0.55} \\ \bottomrule
    \end{tabular}
    \label{tab:trunk_results}
\end{table}

\printbibliography

\iftoggle{ext}{

\appendix
\subsection{Proof of \Cref{lem:wmap_lipschitz}}\label{app:proof_wmap_lipschitz}
\begin{proof}
    Note that $\wmap^{-1}$ exists since $\wmap$ is a local diffeomorphism in $\tubular$. Since $\wmap^{-1}$ is smooth on $\tubular$ and $\tubular$ can be chosen small enough such that it is guaranteed to be a compact set (\eg by shrinking $\epsilon$), then the Jacobian $\jac \wmap^{-1}$ is bounded on $\tubular$. Thus,
    \begin{align*}
        \norm{\jac \wmap^{-1}(\fstate)} \leq \sigma_{\text{max}}(\jac \wmap^{-1}) = \frac{1}{\sigma_{\text{min}}(\jac \wmap)}.
    \end{align*}
    By the Mean Value Theorem, we have that
    \begin{align*}
        \norm{
\begin{bmatrix} \rstate - \rstate' \\ \nstate - \nstate' \end{bmatrix}
} \leq \frac{1}{\sigma_{\text{min}}(\jac \wmap)} \norm{\wmap(\rstate, \nstate) - \wmap(\rstate', \nstate')},
\end{align*}
which by rearranging, gives us the lower bound. Finally, since $\wmap$ is differentiable on a convex subset of $(\chart \times \normalchart)\big(\normalbundleeps\big)$, then we can apply the mean value theorem such that 
\begin{equation}
\begin{aligned}
    \norm{\wmap(\rstate, \nstate) - \wmap(\rstate', \nstate')} &\leq \norm{\jac\wmap} \norm{\begin{bmatrix} \rstate - \rstate' \\ \nstate - \nstate' \end{bmatrix}} \\
    &\leq \sigma_{\text{max}} (\jac \wmap) \norm{\begin{bmatrix} \rstate - \rstate' \\ \nstate - \nstate' \end{bmatrix}},
\end{aligned}
\end{equation}
which gives us the upper bound and completes the proof.
\end{proof}

\subsection{Proof of \Cref{corr:asymptotic_rate_foliation}}\label{app:proof_foliation_rate}
\begin{proof}
We begin by telescoping the denominator,
\begin{equation*}
\scriptsize
\begin{aligned}
\lim_{t \to \infty}\frac{\norm{\flow(\qq) - \flow(\pp)}}{\norm{\flow(\qq) - \flow(\pp')}} &= \lim_{t \to \infty}\frac{\norm{\flow(\qq) - \flow(\pp')}}{\norm{\flow(\qq) - \flow(\pp) + \flow(\pp) - \flow(\pp)}} \\
&\leq \lim_{t \to \infty} \frac{1}{\Big| \frac{\norm{\flow(\qq) - \flow(\pp)}}{\norm{\flow(\qq) - \flow(\pp)}} - \frac{\norm{\flow(\pp) - \flow(\pp')}}{\norm{\flow(\qq) - \flow(\pp')}}\Big|} \\
&= \lim_{t \to \infty} \frac{1}{\Big| 1 - \frac{\norm{\flow(\pp) - \flow(\pp')}}{\norm{\flow(\qq) - \flow(\pp)}}\Big|} ,
\end{aligned}
\end{equation*}
where the second inequality is due to the triangle inequality. By the invariant foliation theorem (\Cref{thm:foliations}(iii)) and since $\pp' \in \mathcal{F}(\pp')$, we have that the ratio $\frac{\norm{\flow(\pp') - \flow(\pp)}}{\norm{\flow(\qq) - \flow(\pp)}} \to \infty$ as $t \to \infty$. Hence,
\begin{align*}
    \lim_{t \to \infty} \frac{\norm{\flow(\qq) - \flow(\pp)}}{\norm{\flow(\qq) - \flow(\pp')}} \leq \lim_{t \to \infty} \frac{1}{\Big| 1 - \frac{\norm{\flow(\pp) -\flow(\pp')}}{\norm{\flow(\qq) - \flow(\pp)}}\Big|} = 0 .
\end{align*}
\end{proof}

\subsection{Proof of \Cref{lemma:flow_estimates}}\label{app:proof_flow_estimates}
\begin{proof} Since $\qq \in \mathcal{F}(\pp)$, we can denote the reduced coordinates as a function of the normal coordinates, \ie $\rstatebar = \chart(\pp), \rstate' = \chart(\pp')$.
The constant $\gamma_2 > 0$ is chosen such that, for any $\pp'$ satisfying $\norm{\pp - \pp'} < \gamma_2$, we have that $\norm{\rstate' - \rstatebar}^2 < \eta_1 \norm{\rstate' - \rstatebar}$ for an appropriately chosen $\eta_1 > 0$.

We derive the upper bound on the flow in the reduced coordinates by Taylor expanding $\varphi_t$ at $\rstate'$,
\begin{equation}
\small
\begin{aligned}
    \varphi_t(\rstatebar, \zero) = \varphi_t(\rstate', \zero) &+ \jac_1 \varphi_t(\rstate', \zero) (\rstate' - \rstatebar) \\
    &+ \bigO(\norm{\rstate' - \rstatebar}^2),
\end{aligned}
\end{equation}
where we used the fact that $\jac_2 \varphi_t(\rstate', \zero) = \zero$.
Using this expansion, we construct the following bound
\begin{equation} \label{eq:lemmabound}
\small
\begin{aligned}
    &\norm{\varphi_t(\rstatebar, \zero) - \varphi_t(\rstate', \zero)} \\
    & \hspace{2em}= \norm{\big(\jac_1 \varphi_t(\rstate', \zero) (\rstate' - \rstatebar) + \bigO(\norm{\rstate' - \rstatebar}^2)} \\
    &\hspace{2em}\leq \norm{\jac_1 \varphi_t(\rstate', \zero)}\norm{\rstate' - \rstatebar} + C_1 e^{-\mu t} \norm{\rstate' - \rstatebar}^2 \\
    &\hspace{2em}\leq (C + C_1\eta_1)e^{-\mu t}\norm{\rstate' - \rstatebar},
\end{aligned}
\end{equation}
where the first inequality results from the fact that under Lemma C.1 in \cite{eldering2013normally}, if $\norm{\jac_1 \varphi_t} \leq Ce^{-\mu t}$, then there exists $\bar{C}$ such that $\norm{\jac^k_1 \varphi_t} \leq \bar{C}e^{-\mu t}$ for $k \leq r$ and the integral remainder form of a Taylor series \cite{marsden1993elementary} (Theorem 9, p. 195) is bounded by
\begin{equation}
\small
\begin{aligned}
    \bigO(\norm{\rstate' - \rstatebar}^2) &\leq \frac{1}{2} \norm{\jac^2_1 \varphi_t} \norm{\rstate' - \rstatebar}^2 \\
    &\leq \frac{1}{2} \bar{C} e^{-\mu t} \eta_1 \norm{\rstate' - \rstatebar}.
\end{aligned}
\end{equation}
In Equation \eqref{eq:lemmabound}, we let $C_1 = \frac{1}{2}\bar{C}$. Notice that we used \Cref{defn:ssm}(iii), and the fact that $\gamma_2$ is small enough to ensure $\norm{\rstate' - \rstatebar}^2 < \eta_1 \norm{\rstate' - \rstatebar}$. The lemma follows with $K=C + C_1\eta_1$. 
\end{proof}

\subsection{Proof of \Cref{prop:projection}}\label{app:proof_projection_operator}
\begin{proof}
    By \Cref{thm:foliations}, $\fiber(\cdot;\pp)$ is a $C^r$-smooth surface with the key property that its dependence on the base point $\pp$ is also $C^r$ smooth. 
    Expanding $\fiber$ in $\nstate$ around $\nstatebar = \zero$ for fixed $\pp$, we have
    \begin{equation} \label{eq:taylor_base}
    \small
    \begin{aligned}
        \rstate &= \fiber(\nstate; \pp) = \fiber(\zero; \pp) + \jac_{\nstate} \fiber(\zero; \pp) \nstate + \bigO(\norm{\nstate}^2) \\
        &= \rstatebar + \jac_{\nstate} \fiber(\zero; \pp)\nstate + \bigO(\norm{\nstate}^2) \\
        &= \rstatebar + \jac_{\nstate} \fiber \big(\zero; (\rstatebar, \mathbf{g}(\rstatebar)) \big)\nstate + \bigO(\norm{\nstate}^2),
    \end{aligned}
    \end{equation}
    where the first equality is due to the definition of a fiber \eqref{defn:fiber} and the second equality is due to the SSM being constructed as a graph, where $\pp = (\rstatebar, \graphfunc(\rstatebar))$.

    Expanding $\fiber(\cdot; (\rstatebar, \graphfunc(\rstatebar)))$ around $\rstatebar = \zero$ we have
    \begin{equation} \label{eq:taylor_fiber}
    \small
    \begin{aligned}
        \fiber \big(\nstate; (\rstatebar, \graphfunc(\rstatebar)) \big) =&~ \fiber \big(\nstate; (\zero, \zero) \big) \\&+ \jac_{\rstate} \fiber(\nstate; (\zero, \zero))\rstatebar \\&+ \bigO(\norm{\rstatebar}^2),
    \end{aligned}
    \end{equation}
    where we invoke tangency of the manifold at the origin, hence $\jac_{\rstate} \graphfunc(\zero) = \zero$ and $\graphfunc(\zero) = \zero$. Combining \eqref{eq:taylor_base} and \eqref{eq:taylor_fiber}, yields the following
    \begin{equation*}
    \begin{bmatrix}
            \rstatebar \\ \zero
    \end{bmatrix}
    =
    \begin{aligned}
        \begin{bmatrix}
            \rstate - \Vconst \nstate - \Vlin(\rstatebar) \nstate \\
            \zero
        \end{bmatrix},
    \end{aligned}
    \end{equation*}
    where $\Vlin(\rstatebar) = \jac_{\nstate} \jac_{\rstate} \fiber(\zero; \zero) \rstatebar$.
    We take a zeroth-order approximation and define the candidate projection operator
    \begin{equation*}
    \begin{aligned}
        \Pi(\fstate) 
        =
        \begin{bmatrix}
            \rstatebar \\ \zero
        \end{bmatrix}
        =
        \underbrace{\begin{bmatrix}
            \eye_{\nrstate \times \nrstate} \\ \zero_{(\nfstate - \nrstate) \times \nrstate}
        \end{bmatrix}}_{\Vpca}
        \underbrace{\begin{bmatrix}
            \eye_{\nrstate \times \nrstate} & -\Vconst
        \end{bmatrix}}_{\Vopt^\top}
        \begin{bmatrix}
            \rstate \\ \nstate
        \end{bmatrix}
    \end{aligned}.
    \end{equation*}
    Note that $\Vpca$ and $\Vopt$ are not unique since we can always apply a similarity transform to these matrices. 
\end{proof}

\subsection{Reduced Order Model Predictive Control}\label{app:mpc_formulation}
This section describes how we implement trajectory tracking using MPC with our learned dynamics model. 
We formulate the MPC optimization problem as follows
\begin{equation}\label{eq:MPC}
    \begin{aligned}
    \underset{\mathbf{u}(\cdot)}{\operatorname{min}} \quad & \left\|\delta \mathbf{z}\left(t_f\right)\right\|_{\mathbf{Q}_f}^2+\int_{t_0}^{t_f}\left(\|\delta \mathbf{z}(t)\|_{\mathbf{Q}}^2+\|\mathbf{u}(t)\|_{\mathbf{R_u}}^2 \right. \\
    & + \left. \|\mathbf{u}(t)-\mathbf{u}(t-\Delta t)\|_{\mathbf{R}_{\Delta}}^2\right) dt \\
    \text{s.t.} \quad & \mathbf{x}_r (0) =\Vopt^\top \left(\mathbf{y}(0)-\mathbf{y}_{\mathrm{eq}}\right) \\
    & \dot{\mathbf{x}}_r(t) =\mathbf{R} \mathbf{x}_r (t)^{1:n_r} + \mathbf{B}_r \mathbf{u}(t) \\
    & \mathbf{z}(t) = \mathbf{C} \left( \Vpca \rstate (t) + \mathbf{W}_\mathrm{nl} \rstate (t)^{2:n_w} \right) + \mathbf{z}_{\mathrm{eq}}, \\
    & \mathbf{z}(t) \in \mathcal{Z}, \quad \mathbf{u}(t) \in \mathcal{U},
    \end{aligned}
\end{equation}
where $\mathbf{z} \in \R^o$, with $o \leq p$, is the performance state vector that we aim to control.
Equilibrium quantities are denoted by the subscript 'eq', as in $y_\mathrm{eq}$ and $z_\mathrm{eq}$.
Similar to \cite{AloraCenedeseEtAl2023}, we discretize the continuous-time dynamics and solve \eqref{eq:MPC} using GuSTO \cite{BonalliCauligiEtAl2019}.
Running this optimization is designed to be real-time feasible, as $n \ll n_f$.

}%

\end{document}